\documentclass[10pt]{article}
\usepackage{graphicx}
\usepackage{subfigure}
\usepackage{float}
\usepackage{amsfonts,amsbsy,amssymb,amsmath,amsthm,amsfonts,mathtools}
\usepackage{verbatim}
\usepackage{parskip}
\usepackage[T1]{fontenc}
\usepackage{multicol}
\usepackage[onehalfspacing]{setspace}
\usepackage{color}
\usepackage[autostyle]{csquotes}
\usepackage{pgf}
\usepackage{hyperref}
\usepackage{thmtools}
\usepackage{caption}
\usepackage{soul}
\usepackage{pdflscape}
\usepackage{relsize}
 
\declaretheoremstyle[%
  spaceabove=-6pt,%
  spacebelow=6pt,%
  headfont=\normalfont\itshape,%
  postheadspace=1em,%
  qed=\qedsymbol%
]{mystyle} 

\hypersetup{
    colorlinks=true,
    linkcolor=blue,
    filecolor=blue,      
    urlcolor=blue,
    citecolor=blue,
    pdftitle={Overleaf Example},
    pdfpagemode=FullScreen,
    }

\newtheorem{thm}{Theorem}[section]
\newtheorem{prop}{Proposition}[section]
\newtheorem{cor}[thm]{Corollary}

\newtheorem{lemma}[thm]{Lemma}

\newtheorem{example}[thm]{Example}

\newcommand{\FF}{{\rm I\!F}}

\newcommand{\CC}{\mathcal{C}}
\newcommand{\CD}{\mathcal{D}}

\newcommand{\ba}{\mathbf{a}}

\title{Additive Complementary Pairs of Codes}

\author{Sanjit Bhowmick$^{1}$, Deepak Kumar Dalai$^{1}$  
\footnote{
$^{1}$School of Mathematical Sciences,
National Institute of Science Education and Research,\\
An OCC of Homi Bhabha National Institute, Bhubaneswar, Odisha 752050, India. Email: sanjitbhowmick@niser.ac.in; deepak@niser.ac.in\\
}
}

\begin{document}
\maketitle
\begin{abstract} 
An additive code is an $\FF_q$-linear subspace of $\FF_{q^m}^n$ over $\FF_{q^m}$, which is not a linear subspace over $\FF_{q^m}$. Linear complementary pairs (LCP) of codes have important roles in cryptography, such as increasing the speed and capacity of digital communication and strengthening security by improving the encryption necessities to resist cryptanalytic attacks. This paper studies an algebraic structure of additive complementary pairs (ACP) of codes over $\FF_{q^m}$. Further, we characterize an ACP of codes in analogous generator matrices and parity check matrices. Additionally, we identify a necessary condition for an ACP of codes. Besides, we present some constructions of an ACP of codes over $\FF_{q^m}$ from LCP codes over $\FF_{q^m}$ and also from an LCP of codes over $\FF_q$. Finally, we study the constacyclic ACP of codes over $\FF_{q^m}$ and the counting of the constacyclic ACP of codes.
\end{abstract}


\noindent\textbf{2020 AMS Classification Code:} 94B05. 

\section{Introduction}\label{sec:intr}
Linear complementary pairs (LCP) of codes, which were introduced by Bhasin et al. in~\cite{NBD15}, are extensively explored for wide application in cryptography (see~\cite{BCC14,CG16}). Further, we refer to some papers~\cite{BDM23,CG18,GOS18} on the LCP of codes over a finite field, and we point out that complementary pairs of codes without linearity properties have an essential role in constructing quantum codes (see~\cite{Rains99}). Thus, it is a useful context for studying non-linear complementary pairs of codes in coding theory. In particular, we study additive complementary pairs (ACP) of codes that extend the analogous of additive complementary dual (ACD) codes. With this point of view, lots of research has been done for additive complementary dual (ACD) codes (see~\cite{shi20,shi22, shi222}).
LCP of codes is a generalization of linear complementary dual (LCD) codes. The notion of LCD codes was introduced by Massy in~\cite{Mas92} and later studied in~\cite{CMT18,CMTQ18,CMTQ19,GOS18,MTQ18}.

Besides, the study of additive codes has gathered significant attention due to their theoretical significance and practical applications. One prominent application of additive codes lies in quantum error correction, where they are employed to protect quantum information from errors induced by noise or other environmental factors. Additionally, additive codes find utility in secret-sharing schemes, where they facilitate the secure distribution of confidential information among multiple parties, ensuring that only authorized subsets of participants can reconstruct the original secret~\cite{Kim17}. Additive codes represent a crucial class of codes within coding theory, introduced by Delsarte et al.~\cite{Del73}. Generally, additive codes are subgroups of the underlying abelian group. Further, Huffman~\cite{Huff13} provided an algebraic structure for additive cyclic codes over $\FF_4$. Through ongoing research and exploration, the potential of additive codes continues to expand, offering valuable insights and solutions in the realms of information theory, quantum computing, and cryptography (see~\cite{Ezerman}). Recently, Shi et al.~\cite{shi22} developed a theory of ACD codes over $\FF_4$ for trace Euclidean and Hermitian inner products. The authors introduced a nice construction of ACP codes over $\FF_4$ from binary codes in the same paper. They also established that ACD codes are potentially related to LCD codes. In the same spirit, Choi et al.~\cite{choi23} studied ACP codes over a finite field $\FF_{q^m}$ with $m\geq 2$ in a general setup of trace inner product. They also provided new construction of ACP codes with respect to trace-Euclidean, Hermitian, and Galois inner products. This general setup showed that ACD codes are related to LCD codes. In this same paper, they further computed good numerical examples of ACP codes. Motivated by these papers, we study the theory of additive complementary pair (ACP) of codes. We further derive some construction of ACP of codes. This paper establishes a relationship between the LCP of codes and the ACP of codes.

On the other hand, an $\dfrac{\FF_q[X]}{\left( X^n-\lambda\right)}$-submodule of $\dfrac{\FF_{q^m}[X]}{\left( X^n-\lambda\right)}$ is a $\lambda$-constacyclic $\FF_q$-linear additive code over $\FF_{q^m}$, where $\lambda\in \FF_q^*$ and $m\geq 2$. In paper~\cite{cao15}, the authors studied the algebraic structure of $\lambda$-constacyclic $\FF_q$-linear additive code over $\FF_{q^m}$. In the same paper, they developed a theory for self-orthogonal and self-dual negacyclic $\FF_q$-linear codes over $\FF_{q^l}$. The authors in paper~\cite{shi222} deduced a theory for additive cyclic and cyclic $\FF_2$-linear ACD codes over $\FF_4$ of odd length for the trace Euclidean and Hermitian inner product. They also provided a characterization for subfield subcodes. Inspired by these papers, we will study the constacyclic $\FF_q$-linear ACP of codes and counting of constacyclic $\FF_q$-linear ACP of codes.


The paper is organized as follows. In Section~\ref{sec:prel}, firstly, we sketch basic definitions, notations, and results on $\FF_q$-linear additive codes over $\FF_{q^m}$, and secondly, we recall some useful results for our context from constacyclic $\FF_q$-linear additive codes over $\FF_{q^m}$. Section~\ref{sec:char} deals with the characterization of ACP of codes with respect to the general inner product on $\FF_{q^m}^n$. Some constructions of ACP of codes are presented in Section~\ref{sec:build}.
Further, we count a formula for constacyclic ACP of codes in Section~\ref{sec:consta}. Finally, the paper concludes in Section~\ref{sec:con}.

\section{Some preliminaries}\label{sec:prel}
\subsection{$\FF_q$-linear additive codes}~\label{ssec:Fqlinear}
Let $\FF_q$ and $\FF_{q^m}$ be the finite fields with cardinality $q$ and $q^m$, respectively, and characteristic $p$. A nonempty subset $\CC$ of $\FF_{q^m}^n$ (where $m > 1$) is called an $\FF_q$-linear additive code if $\CC$ is an $\FF_q$-linear subspace of $\FF_{q^m}^n$ (see for more details \cite{choi23}). Note that an $\FF_q$-linear additive code of $\FF_{q^m}^n$ is an $\FF_q$-linear code over $\FF_{q^m}$ of length $n$. Now, we define a bilinear mapping 
\begin{align}\label{eq-1.11}
\mathcal{B}~:& ~\FF_{q^m}^n\times \FF_{q^m}^n\rightarrow \FF_{q} \text{ such that} & \nonumber \\
((a_1,a_2,\ldots,a_n),(b_1,b_2,\ldots,b_n))& \mapsto \mathcal{B}((a_1,a_2,\ldots,a_n),(b_1,b_2,\ldots,b_n))=\sum\limits_{i=1}^n Tr(\mu_ia_i\pi(b_{\sigma(i)})) &
\end{align}
where $Tr : \FF_{q^m}\rightarrow \FF_q$ is the trace mapping defined by $x\mapsto x+x^q+\cdots+x^{q^{m-1}}$, $\pi$ is a field automorphism on $\FF_{q^m}$ and $\sigma~:~\{1,2,\ldots,n\}\rightarrow\{1,2,\ldots,n\}$ is a permutation with corresponding matrix $P$ such that
\[
P_{ij} = \left\{\begin{array}{ll}
1 & \textrm{if }i=\sigma( j );\\
0 & \textrm{otherwise,}
\end{array}\right.
\]
$M=diag(\mu_1,\mu_2,\ldots,\mu_n)$ is an $n\times n$ matrix over $\FF_{q^m}$ with $\mu_i\in\FF_{q^m}^*,$ and $(a_1,a_2,\ldots,a_n),(b_1,b_2,\ldots,b_n)\in \FF_{q^m}^n.$
Note that $\mathcal{B}$ is not symmetric in general. $\mathcal{B}$ satisfies the following property
\begin{enumerate}
    \item[(a)] $\mathcal{B}$ is non-degenerate;
    \item[(b)] $Rad_{L}(B)=\{\ba\in\FF_{q^m}^n~|~\mathcal{B}(\ba,\mathbf x)=0~\text{for~all}~\mathbf x\in \FF_{q^m}^n\}$ is an $\FF_q$-linear subspace of $\FF_{q^m}^n$;
    \item[(c)] $Rad_{R}(B)=\{\ba\in\FF_{q^m}^n~|~\mathcal{B}(\mathbf y,\ba)=0~\text{for~all}~\mathbf y\in \FF_{q^m}^n\}$ is an $\FF_q$-linear subspace of $\FF_{q^m}^n$.
\end{enumerate}
For an $\FF_q$-linear additive code $\CC$ over $\FF_{q^m}$, the left-dual of $\CC$ is denoted as $\CC^{\perp_L}$, defined by $$\CC^{\perp_L}=\{\ba\in\FF_{q^m}^n~|~\mathcal{B}(\ba,\mathbf c)=0~\text{for~all}~ \mathbf c\in \CC\}$$ 
and the right-dual of $\CC$ is denoted as $\CC^{\perp_R}$, defined by $$\CC^{\perp_R}=\{\ba\in\FF_{q^m}^n~|~\mathcal{B}(\mathbf c,\ba)=0~\text{for~all}~\mathbf c\in \CC\}.$$
Note that $\CC^{\perp_L}$ and $\CC^{\perp_R}$ are $\FF_q$-linear subspace of $\FF_{q^m}^n$.
It is not difficult to show that $\CC^{\perp_L}$ and $\CC^{\perp_R}$ satisfy the  conditions (for more details, we refer~\cite{choi23})
\begin{eqnarray}\label{eq}
&& \dim_{\FF_q}(\CC)+\dim_{\FF_q}(\CC^{\perp_L})=nm \label{eq-1.13a}\\
&& \dim_{\FF_q}(\CC)+\dim_{\FF_q}(\CC^{\perp_R})=nm \label{eq-1.13b}\\
&&(\CC^{\perp_L})^{\perp_R}=(\CC^{\perp_R})^{\perp_L}=\CC. \label{eq-1.13}
\end{eqnarray}
\subsection{Constacyclic $\FF_q$-linear additive codes}\label{ssec:constaFqlinear}
In this subsection, we assume that $\gcd(n,p)=1$. Suppose that $\lambda \in \FF_{q}^*$ with multiplicative order $t$. Consequently, $\gcd(nt,p)=1$ as $t$ is divisor of $q-1$. Now, we spell out some notations for our contexts as follows.
\begin{itemize}
\item Set $N=nt$.
\item Denote $\CC_i^{(b)}=\{i,ib,ib^2,\ldots\}\pmod{N}$ as the $b$-cyclotomic coset containing $i$ modulo $N$ where $i\in\{0,1,\ldots,n-1\}$ and $b$ is either $q$ or $q^m$.
\item Denote the cardinality of $\CC_i^{(b)}$ as $|\CC_i^{(b)}|$.
\end{itemize}
We following lemma is presented to enrich our context.
\begin{lemma}\cite[Lemma 2.1]{cao15}
\begin{itemize}
\item[$(a)$] If $\gcd(|\CC_{i}^{(q)}|,m)=1$, then $\CC_{i}^{(q)}=C_{i}^{(q^m)}$.
\item[$(b)$] If $m$ is a factor of $|\CC^{(q)}_{i}|$, then $|\CC_{i}^{(q)}|=m|\CC_{i}^{(q^m)}|$ and $\CC_{i}^{(q)}=\CC_{i}^{(q^m)}\cup \CC_{iq}^{(q^m)}\cup \ldots \cup \CC_{iq^{m-1}}^{(q^m)}$, where $\CC_{iq^j}^{(q^m)}, 0 \leq j \leq m-1 $ are pairwise disjoint with same cardinality.
\end{itemize}   
\end{lemma}
Let $X$ be an indeterminate over $\FF_{q^m}$. Consider the quotient ring 
\begin{align*}
\mathcal{R}^{(q)}_{n,\lambda}&=\dfrac{\FF_{q}[X]}{(X^n-\lambda)}=\bigoplus_{i=0}^{s}\dfrac{\FF_q[X]}{(p_i(X))}=\bigoplus_{i=0}^{s} \mathcal{K}_i& 
\end{align*}
where $X^n-\lambda=\prod\limits_{i=0}^{s}p_i(X)$ with each $p_i(X)$ is an irreducible polynomial in $\FF_q[X]$, corresponding $q$-cyclotomic coset $\CC^{(q)}_{1+tj}$ and $\mathcal{K}_j=\dfrac{\FF_q[X]}{(p_j(X))}$ for $0\leq j \leq s$.  
Further, consider the quotient ring 
\begin{align*}
\mathcal{R}^{(q^m)}_{n,\lambda}&=\dfrac{\FF_{q^m}[X]}{(X^n-\lambda)}=\bigoplus_{i=0}^{\rho}\dfrac{\FF_{q^m}[X]}{(M_i(X))}=\bigoplus_{i=0}^{\rho} \mathcal{I}_i& 
\end{align*}
where $X^n-\lambda=\prod\limits_{i=0}^{\rho}M_i(X)$ with each $M_i(X)$ is a polynomial in $\FF_{q^m}[X]$ and $\mathcal{I}_j=\dfrac{\FF_{q^m}[X]}{(M_j(X))}$ for $0\leq j \leq \rho$ with $\rho=r+(s-r)m=ms-(m-1)r$ and 
\[M_{i}(X) = \left\{\begin{array}{ll}
q_i(X) & \textrm{if }~0\leq i\leq r;\\
q_{r+(k-1)m+1}(X)q_{r+(k-1)m+2}(X)\cdots q_{r+km}(X) & \textrm{if}~1\leq k\leq s-r.
\end{array}\right.\]
Here each $q_i(X)$ is an irreducible polynomial in $\FF_{q^m}[X]$, corresponding $q^m$-cyclotomic coset $\CC^{(q)}_{1+tj}$  for $0\leq j \leq s$. 
Denote $\mathcal{F}_i=\oplus_{h=1}^{m}\mathcal{I}_{r+(i-1)m+h}$, for $1\leq i\leq s-r$. For more details, see~\cite[Section $2$]{cao15}.

We next recall that an $\FF_q$-linear additive code of length $n$ over $\FF_{q^m}$ is defined as an $\FF_q$-linear subspace of $\FF_{q^m}^n$. In view of the above discussion, we have the following:
\begin{prop}\label{prop-1.1}
\begin{enumerate}
\item[$1)$] Each nonempty $\lambda$-constacyclic $\FF_q$-linear additive code $\CC$ in $\FF_{q^m}^n$ can be uniquely expressed as 
\begin{equation}\label{eq-11qq}
\CC=\CC_1\bigoplus \CC_2\bigoplus\cdots\bigoplus \CC_s,
\end{equation} 
where each $\CC_i$ is a unique $\mathcal{K}_i$-subspace of $\mathcal{I}_i$, for $0\leq i \leq r$ and $\CC_{r+k}$ is a unique $\mathcal{K}_i$-subspace of $\mathcal{F}_{r+k}$ for $1\leq k \leq s-r$.  
\item[$2)$] Conversely, let $\CD_i$ be an $\mathcal{K}_i$-subspace of $\mathcal{I}_i$, for $0\leq i \leq r$ and $\CD_{r+k}$ be an $\mathcal{K}_i$-subspace of $\mathcal{F}_{r+k}$ for $1\leq k \leq s-r$. Then, the direct sum 
\begin{equation}\label{eq-21qq}
\CD=\CD_1\bigoplus \CD_2\bigoplus\cdots\bigoplus \CD_s
\end{equation}
is an $\lambda$-constacyclic $\FF_q$-linear additive code over $\FF_{q^m}$.
\end{enumerate}
\end{prop}
\begin{proof}
The proof of Proposition $1$ follows from~\cite[Lemma 3.1]{cao15}, and the proof of Proposition $2$ is a straightforward exercise.  
\end{proof}
\section{Characterization of ACP of codes}\label{sec:char} 
For two $\FF_q$-linear additive codes over $\FF_{q^m}$, $\CC$ and $\CD$, the pair $(\CC, \CD)$ is called an additive complementary pair (ACP) of codes if $\CC\oplus_{\FF_q} \CD=\FF_{q^m}^n$, i.e., $\CC\cap \CD=\{0\}$ and $\CC +_{\FF_q}\CD = \{\alpha \mathbf{c} + \beta \mathbf{d} | \alpha, \beta \in  \FF_q, \mathbf{c} \in \CC, \mathbf{d} \in \CD \} = \FF_{q^m}^n$. Equivalently, a pair $(\CC,\CD)$ is an ACP of codes if and only if $\CC\cap \CD=\{0\}$ and $\dim_{\FF_q}(\CC)+\dim_{\FF_q}(\CD)=nm$. Then, we have the following results.
\begin{lemma}\label{lm-1.1}
Let $\CC$ and $\CD$ be two $\FF_q$-linear subspaces of $\FF_{q^m}^n$. With respect to the inner product $B$ (see Equation~\eqref{eq-1.11}), we have 
 \begin{enumerate}
\item[$1)$] $(\CC+\CD)^{\perp_L}=\CC^{\perp_L}\cap \CD^{\perp_L}$;
\item[$2)$] $(\CC+\CD)^{\perp_R}=\CC^{\perp_R}\cap \CD^{\perp_R}$;
\item[$3)$] $\CC^{\perp_L}+\CD^{\perp_L}=(\CC\cap \CD)^{\perp_L}$;
\item[$4)$] $\CC^{\perp_R}+\CD^{\perp_R}=(\CC\cap \CD)^{\perp_R}$.
 \end{enumerate}
\end{lemma}
\begin{proof}
\begin{enumerate}
\item[$1)$] Let $x\in (\CC+\CD)^{\perp_L}$. Then 
$\mathcal{B}(\mathbf x,\ba)=0 \text{ for all } \ba\in \CC+\CD \implies  \mathcal{B}(\mathbf x,\mathbf{c+d})=0 \text{ for all } \mathbf c\in \CC, \mathbf d\in \CD.$
If $\mathbf d = 0$ then $\mathcal{B}(\mathbf x,\mathbf c)=0~\text{for all}~\mathbf c\in \CC\implies \mathbf x\in \CC^{\perp_L}$. Similarly If $\mathbf c = 0$ then $\mathcal{B}(\mathbf x,\mathbf d)=0~\text{for all}~\mathbf d\in \CD,\implies \mathbf x\in \CD^{\perp_L}$. Hence, $\mathbf x\in \CC^{\perp_L}\cap \CD^{\perp_L}.$

For the other direction, let $\mathbf y \in  \CC^{\perp_L}\cap \CD^{\perp_L}$. 
Then $\mathcal{B}(\mathbf y,\mathbf c)=0~\text{for~all}~\mathbf c\in \CC ~\text{and}~\mathcal{B}(\mathbf y,\mathbf d)=0~\text{for~all}~\mathbf d\in \CD$.  That implies, $\mathcal{B}(\mathbf y,\mathbf{c+d})=\mathcal{B}(\mathbf y,\mathbf c)+\mathcal{B}(\mathbf y,\mathbf d)=0~\text{for~all}~\mathbf c\in \CC,~\mathbf d\in \CD$. Hence, $\mathbf y\in (\CC+\CD)^{\perp_L}$.
\item[$2)$] It can be proved as in $1)$.
\item[$3)$] As $\CC^{\perp_L}+\CD^{\perp_L}=\left( (\CC^{\perp_L}+\CD^{\perp_L})^{\perp_R} \right)^{\perp_L}$ (see Equation~\eqref{eq-1.13}), $\CC^{\perp_L}+\CD^{\perp_L} = (\CC\cap \CD)^{\perp_L}$ (follows from 2).
\item[$4)$]   It can be proved as in $3)$.
\end{enumerate}
\end{proof}
\begin{thm}\label{th-1.1}
Let $\CC$ and $\CD$ be two $\FF_q$-linear subspaces of $\FF_{q^m}^n$. With respect to the inner product $\mathcal{B}$ (see Equation \eqref{eq-1.11}), the following are equivalent.
\begin{enumerate}
\item[$1)$] the pair $(\CC,\CD)$ is an ACP of codes;
\item[$2)$] the pair $(\CC^{\perp_L},\CD^{\perp_L})$ is an ACP of codes;
\item[$3)$] the pair $(\CC^{\perp_R},\CD^{\perp_R})$ is an ACP of codes.
\end{enumerate}   
\end{thm}
\begin{proof}
 In order to prove that $1)$, $2)$ and $3)$ are equivalent, it is enough to show that  $1)$ and $2)$ are equivalent.
 Let us assume that the pair $(\CC,\CD)$ is an ACP of codes.
 Then $\CC\cap\CD=\{0\}$ and $\dim_{\FF_q}(\CC)+\dim_{\FF_q}(\CD)=nm$. That is, $\CC+_{\FF_{q}}\CD = \FF_{q^m}^n$. Then using result in $1)$ of Lemma~\ref{lm-1.1}, we have $\CC^{\perp_L}\cap \CD^{\perp_L}=\{0\}$ and since $\dim_{\FF_q}(\CC)+\dim_{\FF_q}(\CD)=nm$ from result in Equation~\eqref{eq-1.13a}, we have $\dim_{\FF_q}(\CC^{\perp_L})+\dim_{\FF_q}(\CD^{\perp_L})=nm$.
 Hence, $(\CC^{\perp_L},\CD^{\perp_L})$ is an ACP of codes.

 Conversely, suppose that the pair $(\CC^{\perp_L},\CD^{\perp_L})$ is an ACP of codes, which gives that $\CC^{\perp_L}\cap \CD^{\perp_L}=\{0\}$ and $\CC^{\perp_L}+_{\FF_q} \CD^{\perp_L}=\FF_{q^m}^n$. From Lemma~\ref{lm-1.1}, it follows that $\CC+_{\FF_q}\CD=\FF_{q^m}^n$ and $\CC\cap \CD=\{0\}$. Hence, $(\CC,\CD)$ is an ACP of codes.
\end{proof}
For a matrix $G = (g_{ij})$ with entries $g_{ij} \in \FF_{q^m}$, we denote the matrix $Tr(G) = (Tr(g_{ij}))$ over $\FF_q$.
The following theorem presents a characterization of an ACP of codes. 
\begin{thm}\label{th-1.2}
If $\CC$ and $\CD$ are two $\FF_q$-linear additive codes over $\FF_{q^m}$ of length $n$ with generator matrices $G_1$ and $G_2$, respectively. If the pair $(\CC,\CD)$ is an ACP of codes, then 
$rank\left(Tr\left( \begin{array}{cc}
M\left(\pi(G_1)P\right)^\top  \\
M\left(\pi(G_2)P\right)^\top 
\end{array}\right)\right)=n$,
where $\pi$ is a field automorphism on $\FF_{q^m}$ and $\sigma~:~\{1,2,\ldots,n\}\rightarrow\{1,2,\ldots,n\}$ is a permutation with corresponding matrix $P$ such that
\[P_{ij} = \left\{\begin{array}{ll}
1 & \textrm{if }i=\sigma( j );\\
0 & \textrm{otherwise,}
\end{array}\right.\]
$M=diag(\mu_1,\mu_2,\ldots,\mu_n)$ is an $n\times n$ matrix over $\FF_{q^m}$ with $\mu_i\in\FF_{q^m}^*$. 
\end{thm}
\begin{proof}
Let assume that 
$rank\left(Tr\left( \begin{array}{cc}
M\left(\pi(G_1)P\right)^\top  \\
M\left(\pi(G_2)P\right)^\top 
\end{array}\right)\right) < n$.
Then there exists nonzero $\mathbf x \in\FF_{q}^n$ such that
$\left(Tr\left( \begin{array}{cc}
M\left(\pi(G_1)P\right)^\top  \\
M\left(\pi(G_2)P\right)^\top 
\end{array}\right)\right)\mathbf x^\top = 0$. 
Then 
$\left(Tr\left( \begin{array}{cc}
M\left(\pi(G_1)P\right)^\top  \\
\end{array}\right)\right)\mathbf x^\top 
= \left(Tr\left( \begin{array}{cc}
M\left(\pi(G_2)P\right)^\top 
 \end{array}\right)\right)\mathbf x^\top=0$.
 That implies $\mathcal{B}(\mathbf x,\mathbf c)=0$ for all $\mathbf c\in \CC$ and $\mathcal{B}(\mathbf x, \mathbf d)=0$ for all $\mathbf d\in \CD$, that is $\mathbf x\in \CC^{\perp_L}$ and $\mathbf x\in \CD^{\perp_L}$. Hence, $\mathbf x(\neq 0)\in \CC^{\perp_L}\cap \CD^{\perp_L}$. Then $(\CC^{\perp_L},\CD^{\perp_L})$ is not an ACP of codes and that implies $(\CC,\CD)$ is not an ACP of codes (see Theorem \ref{th-1.1}). This contradicts the hypothesis that $(\CC,\CD)$ is an ACP of codes. Hence $rank\left(Tr\left( \begin{array}{cc}
      M\left(\pi(G_1)P\right)^\top  \\
      M\left(\pi(G_2)P\right)^\top  \\
 \end{array}\right)\right) = n$.
\end{proof}
\begin{example}
Take $\FF_q=\FF_2$ and $\FF_{q^m}=\FF_{4}$ with $\pi$ is the automorphism on $\FF_{q^m}$, define as $\pi(x)=x^2$. Let $\CC$ and $\CD$ be two $\FF_2$-linear additive codes over $\FF_4$ of length $6$ with generator matrices
$$G_1=\left(\begin{array}{cccccc}
1 & 0& 0 & 0& 0 & 1 \\
\omega & 0& 0 & 0& 0 &\omega \\
0 & 1& 0 & 0& 0 & 1 \\
0 & \omega & 0 & 0& 0 &\omega \\
0 & 0& 1 & 0& 0 & 1 \\
0 & 0 & \omega & 0& 0 &\omega \\
0 & 0& 0 & 1& 0 & 1 \\
0 & 0 & 0 & \omega & 0 &\omega \\
 \end{array}\right) \text{ and } G_2=\left(\begin{array}{cccccc}
0 & 0& 0 & 0& 1 & 0 \\
0 & 0& 0 & 0& \omega & 0 \\
1 & 1& 1 & 1& 0 & 1 \\
\omega & \omega & \omega & \omega& 0 &\omega \\
\end{array}\right),$$
respectively, where $\omega$ is a primitive root of $\FF_4$. 
It is easy to see that $(\CC,\CD)$ is an ACP of codes with respect to $P$ such that
\[ P_{ij} = \left\{\begin{array}{ll}
1 & \textrm{if~ off-diagonal element };\\
0 & \textrm{otherwise,}
\end{array}\right. \]
$M=diag(\mu_1,\mu_2,\ldots,\mu_n)$ is an $n\times n$ matrix over $\FF_{4}$ with $\mu_i\in\FF_{4}^*$. However, 
$rank\left(Tr\left( \begin{array}{cc}
M\left(\pi(G_1)P\right)^\top  \\
M\left(\pi(G_2)P\right)^\top
\end{array}\right)\right) = 6$, which is coincide with Theorem~\ref{th-1.2}.
\end{example}

We can have similar results in terms of parity check matrices as in the following Corollary.
\begin{cor}
If $\CC$ and $\CD$ are two $\FF_q$-linear additive codes over $\FF_{q^m}$of length $n$ with parity check matrices $H_1$ and $H_2$, respectively. If the pair $(\CC,\CD)$ is an ACP of codes then 
$rank\left(Tr\left( \begin{array}{cc}
M\left(\pi(H_1)P\right)^\top  \\
M\left(\pi(H_2)P\right)^\top
\end{array}\right)\right)=n$, where $\pi$ is a field automorphism on $\FF_{q^m}$ and $\sigma~:~\{1,2,\ldots,n\}\rightarrow\{1,2,\ldots,n\}$ is a permutation with corresponding matrix $P$ such that
\[P_{ij} = \left\{\begin{array}{ll}
1 & \textrm{if }i=\sigma( j );\\
0 & \textrm{otherwise,}
\end{array}\right.\]
$M=diag(\mu_1,\mu_2,\ldots,\mu_n)$ is an $n\times n$ matrix over $\FF_{q^m}$ with $\mu_i\in\FF_{q^m}^*$.  
\end{cor}
In general, the converse of Theorem~\ref{th-1.2} is not correct. The following example illustrates it.
\begin{example}
Take $\FF_q=\FF_2$ and $\FF_{q^m}=\FF_{4}$ with $\pi$ is identity automorphism on $\FF_{q^m}$. Let $\CC$ and $\CD$ be two $\FF_2$-linear additive codes over $\FF_4$ of length $3$ with generator matrices
$$G_1=\left(\begin{array}{ccc}
1 & 1& 0 \\
\omega & \omega & 0 \\
\omega^2 & 0 &\omega^2 
\end{array}\right) \text{ and } G_1=\left(\begin{array}{ccc}
1 & 1& 1 \\
\omega & \omega & \omega \\
\omega & \omega & 0
\end{array}\right),$$ respectively, where $\omega$ is a primitive root of $\FF_4$. 
It is easy to see that $C\cap D\neq \{0\}$, as $(\omega,~\omega,~0)\in C\cap D$. Hence $(\CC,\CD)$ is not an ACP of codes. However, $rank\left(Tr\left(\begin{array}{cc}
G_1  \\ G_2 
\end{array}\right)\right)=3$.
\end{example}
Now, we present a necessary and sufficient condition for an ACP of codes. 
\begin{thm}\label{theorem-11q}
Let $\CC$ and $\CD$ be two $\FF_q$-linear additive codes over $\FF_{q^m}$ of length $n$ with generator matrices $G_1$ and $G_2$ and parity check matrices $H_1$ and $H_2$, respectively. Assume that $\dim_{\FF_q}(\CC)+\dim_{\FF_q}(\CD)=mn$. Then the pair $(\CC,\CD)$ is an ACP of codes if and only if
$rank\left(Tr\left(H_2M\left(\pi\left(G_1\right)P\right)^{\top}\right)\right) = rank(G_1)$ and
$rank\left(Tr\left(H_1M\left(\pi\left(G_2\right)P\right)^{\top}\right)\right) = rank(G_2)$ where $M$ and $P$ are defined as in Theorem~\ref{th-1.2}.
\end{thm}
\begin{proof}
Suppose that $rank\left(Tr\left(H_2M\left(\pi\left(G_1\right)P\right)^{\top}\right)\right) = rank(G_1)$  and $rank\left(Tr\left(H_1M\left(\pi\left(G_2\right)P\right)^{\top}\right)\right) = rank(G_2)$. 
To prove $(\CC,\CD)$ is an ACP of codes, we need to show that $\CC\cap \CD=\{0\}$. Let $G_1$ be an $k\times n$ matrix and $G_2$ be an $nm-k\times n$ matrix. If $\mathbf x\in \CC\cap \CD$ then
\begin{align*}
&\mathbf x=\mathbf{\alpha} G_1~\text{and}~\mathbf x= \mathbf{\beta} G_2~\text{where}~\mathbf{\alpha}\in \FF_q^n,~\mathbf{\beta}\in \FF_q^{nm-k},&\\
\implies& \pi(\mathbf{\alpha})\pi(G_1)=\pi(\mathbf{\beta})\pi(G_2),~\text{where}~ \pi ~\text{is a field automorphism on}~ \FF_{q^m},&\\
\implies&\pi(\mathbf{\alpha})\pi(G_1)P=\pi(\mathbf{\beta})\pi(G_2)P,~\text{where}~ \sigma~:~\{1,2,\ldots,n\}\rightarrow\{1,2,\ldots,n\}~ &\\
&\text{ is a permutation with corresponding matrix } P \text{ such that }
P_{ij} = \left\{\begin{array}{ll}
1 & \textrm{if }i=\sigma( j );\\
0 & \textrm{otherwise,}
\end{array}\right.
& \\
\implies & M\left(\pi(G_1)P\right)^\top\pi(\alpha)^\top = M\left(\pi(G_2)P\right)^\top\pi(\beta)^\top,&\\
&\text{where } M = diag(\mu_1,\mu_2,\ldots,\mu_n) \text{ is an } n\times n \text{ matrix over } \FF_{q^m} \text{ with } \mu_i\in\FF_{q^m}^*,& \\
\implies & H_2M\left(\pi(G_1)P\right)^\top\pi(\alpha)^\top = H_2M\left(\pi(G_2)P\right)^\top\pi(\beta)^\top=0,&\\
& H_1M\left(\pi(G_2)P\right)^\top\pi(\beta)^\top = H_1M\left(\pi(G_1)P\right)^\top\pi(\alpha)^\top=0,&\\
\implies & \pi(\alpha)=0 \text{ as }rank\left(Tr\left(H_2M\left(\pi\left(G_1\right)P\right)^{\top}\right)\right) = rank(G_1),&\\
& \pi(\beta) = 0 \text{ as } rank\left(Tr\left(H_1M\left(\pi\left(G_2\right)P\right)^{\top}\right)\right) = rank(G_2)&.
\end{align*}
Hence, $\mathbf x = 0$ i.e., $\CC\cap \CD=\{0\}$. As by hypothesis $\dim_{\FF_q}(\CC)+\dim_{\FF_q}(\CD)=mn$, the pair $(\CC,\CD)$ is an ACP of codes.

Conversely, let us suppose that the pair $(\CC,\CD)$ is an ACP of codes. Let $rank\left(Tr\left(H_2M\left(\pi\left(G_1\right)P\right)^{\top}\right)\right) < rank(G_1)$.
Then there exists $\mathbf x(\neq 0)\in \FF_{q}^n$ such that $\left(Tr\left(H_2M\left(\pi\left(G_1\right)P\right)^{\top}\right)\right)\mathbf x^\top = 0$. 
That implies, 
$\left(Tr\left(H_2M\left(\pi\left(\mathbf x G_1\right)P\right)^{\top}\right)\right) = 0$. That is, $\mathbf x G_1\in\CD$. As already, $\mathbf x G_1\in \CC$, $\mathbf x G_1(\neq 0)\in \CC \cap \CD$, which contradicts the fact that the $(\CC,\CD)$ is an ACP of codes. Therefore, $rank\left(Tr\left(H_2D\pi\left(G_1\right)P\right)^{\top}\right)=rank(G_1)$.
Similarly, we can prove that $rank\left(Tr\left(H_1M\left(\pi\left(G_2\right)P\right)^{\top}\right)\right)=rank(G_2)$.
\end{proof}

\begin{example}
Take $\FF_{q^m}=\FF_{4}$ with a primitive element $\omega$. Let $\CC$ and $\CD$ be two $\FF_2$-linear additive codes over $\FF_4$ of length $3$ with generator matrices
$G_1=\left(\begin{array}{ccc}
1 & 1& 0 \\
\omega & 0 & \omega \\
0 & \omega &\omega 
\end{array}\right)$ and 
$G_2=\left(\begin{array}{ccc}
1 & 0& 1 \\
1 & 1 & 1 \\
\omega & \omega & \omega
\end{array}\right)$, respectively.
Here, we take $M, \pi$ and $\sigma$ are all identity.\\
Consider a parity check matrices 
$H_1=\left(\begin{array}{ccc}
1 & 1 & 1 \\
\omega & \omega & \omega \\
\omega^2 & \omega^2 & 0
\end{array}\right)$ and 
$H_2=\left(\begin{array}{ccc}
1 & 0& 1 \\
1 & 1 & 0 \\
\omega & 0 & \omega
\end{array}\right)$ of $\CC$ and $\CD$, respectively.
Then 
$H_1G_2^\top=\left(\begin{array}{ccc}
0 & 1 & \omega \\
0 & \omega & \omega^2 \\
\omega^2 & 0 & 0
\end{array}\right)$ and 
$H_2G_1^\top=\left(\begin{array}{ccc}
1 & 0& \omega \\
1 & \omega & \omega \\
\omega & 0 & \omega^2
\end{array}\right)$.
Here, $Tr(H_1G_2^\top)=\left(\begin{array}{ccc}
0 & 0& 1 \\
0 & 1 & 1 \\
1 & 0 & 0  
\end{array}\right)$ and 
$Tr(H_2G_1^\top)=\left(\begin{array}{ccc}
0 & 0& 1 \\
0 & 1 & 1 \\
1 & 0 & 1  
\end{array}\right)$ 
which are having rank $3$.
Hence, $(\CC,\CD)$ is ACP over $\FF_4$ (see Theorem~\ref{theorem-11q}).
\end{example}
\begin{example}
Take $\FF_{q^m}=\FF_{4}$ with a primitive element $\omega$ and an automorphism
$\pi$ on $\FF_{q^m}$ defined as $\pi(x)=x^2$. Let $\CC$ and $\CD$ be two $\FF_2$-linear additive codes over $\FF_4$ of length $6$ with generator matrices \\
$G_1=\left(\begin{array}{cccccc}
1 & 0& 0 & 0& 0 & 1 \\
\omega & 0& 0 & 0& 0 &\omega \\
0 & 1& 0 & 0& 0 & 1 \\
0 & \omega & 0 & 0& 0 &\omega \\
0 & 0& 1 & 0& 0 & 1 \\
0 & 0 & \omega & 0& 0 &\omega \\
0 & 0& 0 & 1& 0 & 1 \\
0 & 0 & 0 & \omega & 0 &\omega
\end{array}\right)$ and 
$G_2=\left(\begin{array}{cccccc}
0 & 0& 0 & 0& 1 & 0 \\
0 & 0& 0 & 0& \omega & 0 \\
1 & 1& 1 & 1& 0 & 1 \\
\omega & \omega & \omega & \omega& 0 &\omega
\end{array}\right)$, respectively.\\ 
It is easy to see that $(\CC,\CD)$ is an ACP of codes with respect to $P$ such that
\[P_{ij} = \left\{\begin{array}{ll}
1 & \textrm{if~ off-diagonal element };\\
0 & \textrm{otherwise,}
\end{array}\right.\]
$M = diag(\mu_1,\mu_2,\ldots,\mu_n)$ is an $n\times n$ matrix with $\mu_i\in\FF_{4}^*$.
However, $rank\left(Tr\left( \begin{array}{cc}
M\left(\pi(G_1)P\right)^\top  \\
\end{array}\right)\right) = 6$, which coincides with Theorem~\ref{theorem-11q}, since $rank\left(Tr\left(H_2M\left(\pi\left(G_1\right)P\right)^{\top}\right)\right) = rank(G_1)$ and \\
$rank\left(Tr\left(H_1M\left(\pi\left(G_2\right)P\right)^{\top}\right)\right) = rank(G_2)$.
\end{example}
\section{Building-up construction for ACP of codes}\label{sec:build}
We first present some constructions of $\FF_q$-linear additive code over $\FF_{q^m}$ by using a linear code over $\FF_{q^m}$. Recall that a linear code over $\FF_{q^m}$ of length $n$ is a subspace of $\FF_{q^m}^n$. Let $\Tilde{G}$ be a $(k\times n)$ generator matrix of a linear code $\Tilde{C}$ over $\FF_{q^m}$. As $\FF_{q^m}$ forms a vector space over $\FF_q$, consider $\{1,\alpha,\alpha^2,\ldots,\alpha^{m-1}\}$, a basis of $\FF_{q^m}$ over $\FF_q$. Denote 
\begin{align}\label{eq-qq}
G=& \left(\begin{array}{cc}
\Tilde{G}  \\
\alpha\Tilde{G} \\
\vdots\\
\alpha^{m-1}\Tilde{G}
\end{array}\right).
\end{align}
It can be shown that $G$ is a generator matrix of an $\FF_{q}$-linear additive code, say $\CC$, corresponding to the linear code $\Tilde{\CC}$ over $\FF_{q^m}$, which is nothing but $\Tilde{\CC}$, i.e., $\CC=\Tilde{\CC}$. Next, we will present a construction of ACP of codes from LCP of codes over $\FF_{q^m}$. Towards the LCP of codes, we define a general inner product on $\FF_{q^m}$, which is the analogue of the inner product $\mathcal{B}$ (see Equation~\eqref{eq-1.11}), using the bilinear mapping 
\begin{align}\label{eq-1.11a}
\Tilde{\mathcal{B}}~:& ~\FF_{q^m}^n\times \FF_{q^m}^n\rightarrow \FF_{q^m} \text{ such that }& \nonumber \\
((a_1,a_2,\ldots,a_n),(b_1,b_2,\ldots,a_n))&\mapsto \Tilde{\mathcal{B}}((a_1,a_2,\ldots,a_n),(b_1,b_2,\ldots,a_n))=\sum\limits_{i=1}^n \mu_ia_i\pi(b_{\sigma(i)})&
\end{align}
where $\mu_i, \sigma$, and $\pi$ are as defined in Equation~\eqref{eq-1.11}.
Notice that, in general, $\Tilde{\mathcal{B}}$ is not symmetric. A pair of codes $(\Tilde{\CC},\Tilde{\CD})$ of length $n$ over $\FF_{q^m}$ is called an LCP of codes if $\Tilde{\CC}\oplus_{\FF_{q^m}}\Tilde{\CD}=\FF_{q^m}^n$. That is, a pair of codes $(\Tilde{\CC},\Tilde{\CD})$ of length $n$ over $\FF_{q^m}$ is LCP if and only if $\Tilde{\CC}\cap\Tilde{\CD}=\{0\}$ and $\dim_{\FF_{q^m}}(\Tilde{\CC})+\dim_{\FF_{q^m}}(\Tilde{\CD})=n$. The following proposition presents a characterization of LCP of codes over $\FF_{q^m}$ with respect to the inner product $\Tilde{\mathcal{B}}$ (see Equation~\eqref{eq-1.11a}).
\begin{prop}\label{prop-qq}
Let $\Tilde{\CC}$ and $\Tilde{\CD}$ be two linear codes of length $n$ over $\FF_{q^m}$ with generator matrices $\Tilde{G}_1$, $\Tilde{G}_2$ and parity check matrices $\Tilde{H}_1$, $\Tilde{H}_2$, respectively. Then, the following are equivalent.
\begin{enumerate}
\item The pair $(\Tilde{\CC},\Tilde{\CD})$ is LCP.
\item The $rank\left( M\left(\begin{array}{cc}
      \pi(\Tilde{G}_1)  \\
      \pi(\Tilde{G}_2) 
\end{array}\right)P\right) = n$.
\item The $rank\left(\left(\Tilde{H}_2M\left(\pi\left(\Tilde{G}_1\right)P\right)^{\top}\right)\right)=rank(\Tilde{G}_1)$ and $rank\left(Tr\left(\Tilde{G}_1M\left(\pi\left(\Tilde{G}_2\right)P\right)^{\top}\right)\right)=rank(\Tilde{G}_2)$. Additionally, $\dim_{\FF_{q^m}}(\Tilde{\CC})+\dim_{\FF_{q^m}}(\Tilde{\CD})=n$. 
 \end{enumerate}
Here, $\pi$ is a field automorphism on $\FF_{q^m}$, $\sigma~:~\{1,2,\ldots,n\}\rightarrow\{1,2,\ldots,n\}$ is a permutation with corresponding matrix $P$ such that
\[P_{ij} = \left\{\begin{array}{ll}
1 & \textrm{if }i=\sigma( j );\\
0 & \textrm{otherwise,}
\end{array}\right. \]
and $M=diag(\mu_1,\mu_2,\ldots,\mu_n)$ is an $n\times n$ matrix over $\FF_{q^m}$ with $\mu_i\in\FF_{q^m}^*$ as defined in Equation~\eqref{eq-1.11}.
\end{prop}
The proof of this Proposition is similar to the proof of the Theorem~\ref{theorem-11q}.
Now, we will present an ACP of codes from a given pair of LCP of codes.
\begin{prop}\label{prop-4q}
Let $\Tilde{\CC}$, $\Tilde{\CD}$ be two linear codes of length $n$ over $\FF_{q^m}$. Let $\CC$, $\CD$ be the codes constructed from $\Tilde{\CC}$, $\Tilde{\CD}$, respectively, as presented in Equation~\eqref{eq-qq}. If the pair $(\Tilde{\CC},\Tilde{\CD})$ is an LCP of codes over $\FF_{q^m}$ then $(\CC,\CD)$ is an ACP of codes over $\FF_{q^m}$.
\end{prop}
\begin{proof}
Let $\Tilde{G}_1$, $\Tilde{G}_2$ be generator matrices of the linear codes $\Tilde{\CC}$, $\Tilde{\CD}$, respectively. Using Equation~\ref{eq-qq}, construct $G_1$ and $G_2$ which are generator matrices of $\CC$, $\CD$, respectively. Since $(\Tilde{\CC},\Tilde{\CD})$ is an LCP of codes, $\dim_{\FF_{q^m}}(\Tilde{\CC})+\dim_{\FF_{q^m}}(\Tilde{\CD})=n$ and that implies $\dim_{\FF_q}(\CC)+\dim_{\FF_q}(\CD) = mn$. We need to prove that $\CC\cap \CD=\{0\}$. On the contrary, assume that there exists a non-zero  $\mathbf x \in \FF_q^n$ such that $\mathbf x\in \CC\cap \CD$. Then $\mathbf x = \gamma G_1 = \delta G_2$ for some $\gamma\in\FF_q^k$, $\delta\in\FF_q^{mn-k}$ where $k = rank(G_1)$ and $mn-k = rank(G_2)$. Then from Equation ~\eqref{eq-qq}, we have that $\mathbf x \in \Tilde{\CC}\cap\Tilde{\CD}$. Since $\mathbf x$ is nonzero vector in $\FF_q^n$, it contradicts that $(\Tilde{\CC},\Tilde{\CD})$ is an LCP. Hence, $(\CC,\CD)$ is an ACP of codes.
\end{proof}
Towards this, for an $\FF_q$-linear additive code $C$ over $\FF_{q^m}$ of length $n$, we define
$$Tr(\CC):=\{Tr(\mathbf c): \mathbf c\in \CC\}.$$
Note that $Tr(\CC)$ is a linear code over $\FF_q$ of length $n$. We call $Tr(\CC)$, a trace code of $\CC$. Further, if $G$ is a generator matrix of $\CC$, then $Tr(G)$ (defined in Section~\ref{sec:char}) is a generator matrix of $Tr(\CC)$. 
\begin{thm}\label{th-1zx}
Let $\CC$ and $\CD$ be two $\FF_q$-linear additive codes over $\FF_{q^m}$ of length $n$. If the pair $(\CC,\CD)$ is an ACP of codes over $\FF_{q^m}$ (using identity $\pi$, $\sigma$ and $M$ in Proposition~\ref{prop-qq}), then $(Tr(\CC),Tr(\CD))$ is an LCP of codes over $\FF_{q}$. In particular, if $\CC$ is an additive complementary dual (ACD) codes over $\FF_{q^m}$ (using identity $\pi$, $\sigma$ and $M$ in Proposition~\ref{prop-qq}), then $Tr(\CC)$ is an LCD codes over $\FF_q$.
\end{thm}
\begin{proof}
Let $G_1$ and $G_2$ be two generator matrices of $\CC$ and $\CD$, respectively. Since, the pair $(\CC,\CD)$ is an ACP of codes over $\FF_{q^m}$ (using identity $\pi$, $\sigma$ and $M$ in Proposition~\ref{prop-qq}), then by Theorem~\ref{th-1.2}, we obtain that
$rank\left(Tr\left( \begin{array}{cc}
G_1  \\ G_2 
\end{array}\right)\right)=n$.
That is, $rank\left( \begin{array}{cc}
Tr(G_1)  \\ Tr(G_2) 
\end{array}\right)=n$.
Since $Tr(G_1)$ and $Tr(G_2)$ are generator matrices of $Tr(\CC)$ and $Tr(\CD)$, respectively, from Proposition~\ref{prop-qq}, we get $(Tr(\CC),Tr(\CD))$ is an LCP of codes over $\FF_{q}$.
\end{proof}

\begin{thm}
Let $\Tilde{\CC}$, $\Tilde{\CD}$ be two linear codes of length $n$ over $\FF_{q^m}$. Let $\CC$, $\CD$ be the codes constructed from $\Tilde{\CC}$, $\Tilde{\CD}$, respectively, as presented in Equation~\eqref{eq-qq}. Then the pair $(\Tilde{\CC},\Tilde{\CD})$ is an LCP of codes over $\FF_{q^m}$ if and only if $(\CC,\CD)$ is an ACP of codes over $\FF_{q^m}$ (using identity $\pi$, $\sigma$ and $M$ in Proposition~\ref{prop-qq}).
\end{thm}
\begin{proof}
Combining Proposition~\ref{prop-4q} and Theorem~\ref{th-1zx}, we have the required result. 
\end{proof}

For $\ba=(a_1,a_2,\ldots,a_n)\in \FF_{q^m}^n$ and a code $\CC$, define $${\ba}\CC:=\{(a_1c_1,a_2c_2,\ldots,a_nc_n)~|~(c_1,c_2,\ldots,c_n)\in \CC)\}.$$ Note that, for $\ba \in \left(\FF_{q^m}^*\right)^n$, ${\ba}\CC$ is an $\FF_q$-linear additive code if $C$ is an $\FF_q$-linear additive code. Recall that a linear code $\Tilde{\CC}:=[n,k,d]$ over $\FF_{q^m}$ is MDS if $d=n-k+1$. 
\begin{thm}\label{thm-qw}
Let $\Tilde{\CC}$, $\Tilde{\CD}$ be two linear codes of length $n$ over $\FF_{q^m}$.
Let $\CC$, $\CD$ be the codes constructed from $\Tilde{\CC}$, $\Tilde{\CD}$, respectively, as presented in Equation~\eqref{eq-qq}.
Assume that at least one of codes $\Tilde{\CC}, \Tilde{\CD}$ is MDS and $\dim_{\FF_{q^m}}(\Tilde{\CC})+\dim_{\FF_{q^m}}(\Tilde{\CD}) = n$. Then there exists $\ba\in \left(\FF_{q^m}^*\right)^n$ such that a pair $({\ba}\CC,\CD)$ is an ACP of codes.   
\end{thm}
\begin{proof}
Since $q^m\geq 4$, by \cite[Theorem 5.4]{BDM23}, there exists $\ba\in \left(\FF_{q^m}^*\right)^n$ such that $(\ba\Tilde{\CC}, \Tilde{\CD})$ is an LCP of codes over $\FF_{q^m}$. By Proposition \ref{prop-4q}, we obtain that  $({\ba}\CC,\CD)$ is an ACP of codes.  
\end{proof}  
Let $\FF_q$ be a finite field with $q\geq 3$, $x$ be transcendental over $\FF_q$ and $\mathcal{P}_k = \{f\in \FF_q[x] : \deg(f) \leq k-1\}$. For $\textbf{b}=(\alpha_1,\alpha_2,\ldots,\alpha_n) \in \FF_q^n$, where $\alpha_i$s are distinct elements in $\FF_q$, let $RS_k(\textbf{b})=\{(f(\alpha_1),f(\alpha_2),\ldots,f(\alpha_n))~|~f\in\mathcal{P}_k\}$ be a $k$-dimensional Reed-Solomon code which is MDS. Then $RS_k(\textbf{b})$ is $[n,k,n-k+1]$ code and $RS_k(\textbf{b})^\perp$ is a $[n,n-k,k+1]$ code.

Denote $\widehat{RS_k(\textbf{b})}$ be the $\FF_{q}$-linear additive code corresponding linear code $RS_k(\textbf{b})$ over $\FF_{q^m}$ and $\widehat{RS_k(\textbf{b})^\perp}$ be the $\FF_{q}$-linear additive code corresponding linear code $RS_k(\textbf{b})^\perp$over $\FF_{q^m}$.
\begin{example}
From above discussion, $RS_k(\textbf{b})=\{(f(\alpha_1),f(\alpha_2),\ldots,f(\alpha_n))~|~f\in\mathcal{P}_k\}$ be a $k$-dimensional Reed-Solomon code. Then by Theorem \ref{thm-qw}, there exists $\ba\in \left(\FF_{q^m}^*\right)^n$ such that a pair $({\ba}\widehat{RS_k(\textbf{b})},\widehat{RS_k(\textbf{b})^\perp})$ is an ACP of codes over $\FF_{q^m}$.
\end{example}
Now, we will present a numerical example. 
\begin{example}
Let $\FF_{25}$ be a finite field, $x$ be transcendental over $\FF_{25}$ and $\mathcal{P}_2 = \{f\in \FF_q[x] : \deg(f) \leq 1\}$. For $\textbf{b}=(1,2,3,4) \in \FF_{25}^4$, let $RS_2(\textbf{b})=\{(f(1),f(2),f(3),f(4))~|~f\in\mathcal{P}_2\}$ be a $2$-dimensional Reed-Solomon code which is MDS. Then $RS_2(\textbf{b})$ is $[4,2,3]$ code and $RS_2(\textbf{b})^\perp$ is a $[4,2,3]$ code with generator matrix $G=\left(\begin{array}{cccc}
1 & 1 & 1 & 1 \\
1 & 2 & 3 & 4 
\end{array}\right)$ and parity check matrix $H=\left(\begin{array}{cccc}
1 & 0 & 2 & 2 \\
0 & 1 & 3 & 1 
\end{array}\right)$.
Set $\ba=(2,1,1,1)$, then $(\ba RS_2(\textbf{b}),RS_2(\textbf{b})^\perp)$ forms an LCP over $\FF_{25}$. It is easy to see that $({\ba}\widehat{RS_2(\textbf{b})},\widehat{RS_2(\textbf{b})^\perp})$ is an ACP of codes over $\FF_{25}$.
\end{example}
Let $\Tilde{C}$ be a linear code over $\FF_{q^m}$ of length $n$ with generator matrix $\Tilde{G}$. We define two different expansions of $\Tilde{\CC}$ as $\Tilde{\CC}_{ex_1}$ and $\Tilde{\CC}_{ex_2}$ generated by the generator matrices
$$\Tilde{G}_{ex_1}=\left(\begin{array}{ccc}
\lambda & P \\
0 & \Tilde{G} 
\end{array}\right) \text{ and }
\Tilde{G}_{ex_2}=\left(\begin{array}{ccc}
P' & \Tilde{G} 
\end{array}\right),$$ 
respectively, where $\lambda\in\FF_{q^m}, P\in\FF^n_{q^m}$ and $P'^\top \in \FF_{q^m}^{k}$.
We now construct an ACP of codes from the expansions of an ACP of codes.
\begin{thm}\label{thm-tt}
Let $\Tilde{\CC}$ and $\Tilde{\CD}$ be two linear codes over $\FF_{q^m}$ such that a pair $(\Tilde{\CC},\Tilde{\CD})$ is an LCP of codes (using identity $\pi$, $\sigma$ and $M$ in Proposition~\ref{prop-qq}).
Then there exist $\lambda\in\FF_{q^m}^*$ and $P\in\FF_{q^m}^n$ such that the pair $(\Tilde{\CC}_{ex_1},\Tilde{\CD}_{ex_2})$ is an LCP of codes over $\FF_{q^m}$.
\end{thm}
\begin{proof}
Let $\Tilde{\CC}$ and $\Tilde{\CD}$ be two linear codes over $\FF_{q^m}$ with generator matrices $\Tilde{G}_{1}$ and $\Tilde{G}_{2}$, respectively. 
Consider the codes $\Tilde{\CC}_{ex_1}$ and $\Tilde{\CD}_{ex_2}$ with generator matrices $\Tilde{G}_{ex_1}=\left(\begin{array}{ccc}
\lambda & P \\
0 & \Tilde{G}_1 
\end{array}\right)$ and $\Tilde{G}_{ex_2}=\left(\begin{array}{ccc}
P'^\top & \Tilde{G}_2 
\end{array}\right)$, respectively, where $\lambda\in\FF_{q^m}$, non-zero $P$ is an $(1\times n)$ matrix over $\FF_{q^m}$ and $P'$ is an $1\times(n-k)$ sub-matrix of $P$. Since, a pair $(\Tilde{\CC},\Tilde{\CD})$ is an LCP of codes over $\FF_{q^m}$, then by Proposition~\ref{prop-qq}, $rank\left(\begin{array}{ccc}
\Tilde{G}_1 \\
\Tilde{G}_2 
\end{array}\right)=n$. Hence, 
\begin{equation}\label{eq-equation}
rank\left(\begin{array}{ccc}
0 & \Tilde{G}_1 \\
P'^\top & \Tilde{G}_2 
\end{array}\right)=n.   
\end{equation} 
Let $\Tilde{C}_1$ be a linear code over $\FF_{q^m}$ generated by 
$\left(\begin{array}{ccc}
0 & \Tilde{G}_1 
\end{array}\right)$.
It is easy to see that $\Tilde{\CC}_1\cap \Tilde{\CD}_{ex_2}=\{0\}$ by Equation~\eqref{eq-equation}. Since, $\Tilde{\CC}_1+\Tilde{\CD}_{ex_2}\subseteq \Tilde{\CC}_{ex_1}+\Tilde{\CD}_{ex_2}$, $\dim_{\FF_{q^m}}(\Tilde{\CC}_1+\Tilde{\CD}_{ex_2})\leq\dim_{\FF_{q^m}}(\Tilde{\CC}_{ex_1}+\Tilde{\CD}_{ex_2})$\\
 $\implies \dim_{\FF_{q^m}}(\Tilde{\CC}_1)+\dim_{\FF_{q^m}}(\Tilde{\CD}_{ex_2})-\dim_{\FF_{q^m}}(\Tilde{\CC}_1\cap \Tilde{\CD}_{ex_2})\leq\dim_{\FF_{q^m}}(\Tilde{\CC}_{ex_1})+\dim_{\FF_{q^m}}(\Tilde{\CD}_{ex_2})-\dim_{\FF_{q^m}}(\Tilde{\CC}_{ex_1}\cap \Tilde{\CD}_{ex_2})\implies\dim_{\FF_{q^m}}(\Tilde{\CC}_{ex_1}\cap \Tilde{\CD}_{ex_2})\leq 1$. 
 
Further, we choose $\lambda \in\FF_{q^m}$ such that 
$\det\left(\begin{array}{ccc}
\Tilde{G}_{ex_1} \\
\Tilde{G}_{ex_2} 
\end{array}\right) = \det\left(\begin{array}{ccc}
\lambda & P \\
0 & \Tilde{G}_1 \\
P'^\top & \Tilde{G}_2 
\end{array}\right) \neq 0$.
Note that there exists such $\lambda$ satisfying the condition. Then from Proposition~\ref{prop-qq}, we have that $\left(\Tilde{\CC}_{ex_1},\Tilde{\CD}_{ex_2}\right)$ is an LCP of codes over $\FF_{q^m}$.
\end{proof}
We denote $\CC_{ex_1}$ be an $\FF_q$-linear additive code over $\FF_{q^m}$ corresponding to a linear code $\Tilde{\CC}_{ex_1}$ over $\FF_{q^m}$ (see Equation~\eqref{eq-qq}).
\begin{cor}\label{cor-qq}
Let $\Tilde{\CC}$ and $\Tilde{\CD}$ be two linear codes over $\FF_{q^m}$ such that a pair $\left(\Tilde{\CC},\Tilde{\CD}\right)$ is an LCP of codes (using identity $\pi$, $\sigma$ and $M$ in Proposition~\ref{prop-qq}). Let $\CC_{ex_1}$, $\CD_{ex_2}$ be the codes constructed from $\Tilde{\CC}_{ex_1}$, $\Tilde{\CD}_{ex_1}$, respectively, as presented in Equation~\eqref{eq-qq}. Then the pair $\left(\CC_{ex_1},\CD_{ex_2}\right)$ is an ACP of codes over $\FF_{q^m}$.
\end{cor}
\begin{proof}
By Theorem~\ref{thm-tt} with combining Proposition~\ref{prop-4q}, we have our results.    
\end{proof}
\begin{example}
Let $\Tilde{\CC}$ and $\Tilde{\CD}$ be two linear codes over $\FF_{8^2}$ with generator matrices 

$G_1=\left(\begin{array}{ccccccc}
1 & 1 & 1 & 1 & 1 & 1 & 1 \\
1 & \omega & \omega^2 & \omega^3 & \omega^4 & \omega^5 & \omega^6 \\
1 & \omega^6 & \omega^5 & \omega^4 & \omega^3 & \omega^2 & \omega \\
1 & \omega^2 & \omega^4 & \omega^6 & \omega & \omega^3 & \omega^5 \\
1 & \omega^5 & \omega^3 & \omega & \omega^6 & \omega^4 & \omega^2 
\end{array}\right)$    and 
$G_2=\left(\begin{array}{ccccccc}
1 & 0 & 1 & \omega^6 & \omega^4 & \omega^4 & \omega^6 \\
0 & 1 & \omega^6 & \omega^4 & \omega^4 & \omega^6 & 1 
\end{array}\right)$, respectively, where $\omega$ is $7$-th root of unity.
It can be shown that $\left(\Tilde{\CC},\Tilde{\CD}\right)$ is an LCP of codes over $\FF_{8^2}$. By Theorem~\ref{thm-tt}, $(\Tilde{\CC}_{ex_1},\Tilde{\CD}_{ex_2})$ is an LCP of codes over $\FF_{8^2}$. However, a generator matrices of $\Tilde{\CC}$ and $\Tilde{\CD}$ of the form 
$G_{ex_1}=\left(\begin{array}{cccccccc}
\omega^5 & \omega^6 & \omega^3 & \omega^4 & \omega & \omega^2 & 1 & 0 \\
0 & 1 & 1 & 1 & 1 & 1 & 1 & 1 \\
0 & 1 & \omega & \omega^2 & \omega^3 & \omega^4 & \omega^5 & \omega^6 \\
0 & 1 & \omega^6 & \omega^5 & \omega^4 & \omega^3 & \omega^2 & \omega \\
0 & 1 & \omega^2 & \omega^4 & \omega^6 & \omega & \omega^3 & \omega^5 \\
0 & 1 & \omega^5 & \omega^3 & \omega & \omega^6 & \omega^4 & \omega^2 
\end{array}\right)$    and 
$G_{ex_2}=\left(\begin{array}{cccccccc}
\omega^6 & 1 & 0 & 1 & \omega^6 & \omega^4 & \omega^4 & \omega^6 \\
\omega^3 & 0 & 1 & \omega^6 & \omega^4 & \omega^4 & \omega^6 & 1 
\end{array}\right)$. It is easy to see that $\left(\CC_{ex_1},\CD_{ex_2}\right)$ is an ACP of codes over $\FF_{8^2}$ which is coincide with Corollary~\ref{cor-qq}.
\end{example}
We further derive a construction of ACP of codes over $\FF_{q^m}$ from given $m$ numbers of LCP of codes over $\FF_{q}$ as follows.
\begin{thm}\label{th-000}
Let $\CC_i : [n, k_i]$ and $\CD_i : [n, n-k_i]$ be two linear code over $\FF_q$, for $0\leq i \leq m-1$. If $\alpha_i$'s are distinct elements of $\FF_{q^m}^*$ such that $\{\alpha_0,\alpha_1,\ldots,\alpha_{m-1}\}$ are linearly independent over $\FF_q$ and $(\CC_i,\CD_i)$ is an LCP for $0\leq i \leq m-1$, then the pair $\left(\sum\limits_{i=0}^{m-1}\alpha_i \CC_i, \sum\limits_{i=0}^{m-1}\alpha_i \CD_i \right)$ is an ACP of codes over $\FF_{q^m}$.
\end{thm}
\begin{proof}
Let $\CC=\sum\limits_{i=0}^{m-1}\alpha_i \CC_i$ and $D=\sum\limits_{i=0}^{m-1}\alpha_i \CD_i$. 
Now we need to show that $\dim_{\FF_q}(\CC)+\dim_{\FF_q}(\CD)=mn$ and $\CC\cap \CD=\{0\}$. From the hypothesis, we have $\dim_{\FF_q}(\CC_i)+\dim_{\FF_q}(\CD_i)=n$ for $0\leq i \leq m-1$. Hence $\dim_{\FF_q}(\CC)+\dim_{\FF_q}(\CD)=mn$. To prove the second condition, take $\mathbf x\in \CC\cap \CD$. Then $\mathbf x=\sum\limits_{i=0}^{m-1}\alpha_i \mathbf c_i$ and $\mathbf x=\sum\limits_{i=0}^{m-1}\alpha_i \mathbf d_i$ for some $\mathbf c_i\in \CC_i$ and $\mathbf d_i \in \CD_i$, where $0\leq i \leq m-1$. Therefore,
\begin{equation}\label{equa-qtr}
    \sum\limits_{i=0}^{m-1}\alpha_i \mathbf c_i=\sum\limits_{i=0}^{m-1}\alpha_i \mathbf d_i \implies \sum\limits_{i=0}^{m-1}\alpha_i\left(\mathbf c_i- \mathbf d_i\right)=0.
\end{equation}
Since $\{\alpha_0,\alpha_1,\ldots,\alpha_{m-1}\}$ are linearly independent over $\FF_q$,  $c_i=d_i$, for all $i : 0 \leq i \leq m-1$. Further as $(\CC_i,\CD_i)$ is an LCP for all $i : 0 \leq i \leq m-1$, $\mathbf c_i =\mathbf d_i = 0$.  Hence, $\CC\cap \CD=\{0\}$. 
\end{proof}
\begin{example}
Let $(\CC_1, \CD_1)$ be a pair of linear codes over $\FF_2$ with generator matrices  $G_1=\left(\begin{array}{cccc}
     1 & 1 & 0 & 1 \\
     0 & 1 & 1 & 1 
 \end{array}\right)$ and $G_2=\left(\begin{array}{cccc}
     1 & 1 & 1 & 0 \\
     1 & 0 & 1 & 1 
 \end{array}\right)$. Also, let $(\CC_2, \CD_2)$ be an another pair of linear codes over $\FF_2$ with generator matrices  $G_3=\left(\begin{array}{cccc}
     1 & 1 & 1 & 0 \\
 \end{array}\right)$ and $G_4=\left(\begin{array}{cccc}
     1 & 1 & 0 & 0 \\
     1 & 0 & 1 & 0 \\
     0 & 0 & 0 & 1 \\
 \end{array}\right)$. Choose, $\omega$ be an element of $\FF_4$ such that $\omega^2+\omega+1=0$ and $\{1,\omega\}$ forms a linearly independent set over $\FF_2$. Now, we consider $\CC$ and $\CD$ with generator matrices $G'=\left(\begin{array}{c}
     G_1 \\
     \omega G_3 
 \end{array}\right)$ and $G''=\left(\begin{array}{c}
     G_2 \\
     \omega G_4  
 \end{array}\right)$. It is easy to see that the pair $(\CC, \CD)$ forms an ACP of codes over $\FF_4$ which coincides with the Theorem~\ref{th-000}. 
\end{example}
We next make the following proposition.
\begin{prop}
Let $\Tilde{\CC}:=[n,k]$, $\Tilde{\CD}:=[n,n-k]$ be two linear codes over $\FF_{q^m}$ with a generator matrix $\Tilde{G}$ of $\Tilde{\CC}$ and a parity check matrix $\Tilde{H}$ of $\Tilde{\CD}$. Suppose that $\Tilde{\CC}_{ex_1}$, $\Tilde{\CD}_{ex_2}$ are two linear codes with a generator matrix $\Tilde{G}_{ex_1}=\left(\begin{array}{cc}
     1 &\mathbf d\\
     0 & \Tilde{G} 
\end{array}\right)$ of $\Tilde{\CC}_{ex_1}$ and a parity check matrix $\Tilde{H}_{ex_2}=\left(\begin{array}{cc}
     1 & \mathbf c\\
     0 & \Tilde{H} 
\end{array}\right)$ of $\Tilde{\CD}_{ex_2}$, where $\mathbf c\in \Tilde{C}_{ex_1}$, $\mathbf d\in\Tilde{\CC}^\top_{ex_1}$. 
Let $\CC$, $\CD$ be $\FF_q$-linear additive codes over $\FF_{q^m}$ corresponding to the linear codes $\Tilde{\CC}_{ex_1}$ and $\Tilde{\CD}_{ex_2}$, respectively over $\FF_{q^m}$ (as in Equation~\eqref{eq-qq})
If the pair $(\Tilde{\CC},\Tilde{\CD})$ is an LCP of codes (using identity $\pi$, $\sigma$ and $M$ in Proposition~\ref{prop-qq}), then $(\CC,\CD)$ is an ACP of codes.    
\end{prop}
\begin{proof}
From Proposition~\ref{prop-4q}, it is enough to show that   $(\Tilde{\CC}_{ex_1},\Tilde{\CD}_{ex_2})$ is an LCP of codes. By the construction of $\Tilde{\CC}_{ex_1}$ and $\Tilde{\CD}_{ex_2}$, we obtain that $\Tilde{\CC}_{ex_1}$ is a linear code of length $n+1$ over $\FF_{q^m}$ with dimension $k+1$ and $\Tilde{\CD}_{ex_2}$ is a linear code of length $n+1$ over $\FF_{q^m}$ with dimension $n-k$. It implies that $\dim\left(\Tilde{\CC}_{ex_1}\right)+\dim\left(\Tilde{\CD}_{ex_2}\right) = n+1$. Now we have to show that $\Tilde{\CC}_{ex_1}\cap \Tilde{\CD}_{ex_2}=\{0\}$.

Let $\mathbf x\in \Tilde{\CC}_{ex_1}\cap \Tilde{\CD}_{ex_2}$. This implies $\mathbf x=\alpha \Tilde{G}_{ex_1}=\beta G_2$, where $\alpha\in \FF_{q^m}^{k+1}$, $\beta\in \FF_{q^m}^{n-k}$ and $G_2$ is a generator matrix of $\Tilde{\CD}_{ex_2}$. This gives that $\alpha \Tilde{G}_{ex_1} \Tilde{H}^\top_{ex_2}=0$ as $G_2\Tilde{H}^\top_{ex_2}=0$. Then
  $$\alpha \Tilde{G}_{ex_1} \Tilde{H}^\top_{ex_2} = 0\implies \alpha\left(\begin{array}{cc}
     1 & \mathbf d\\
     0 & \Tilde{G} 
 \end{array}\right)\left(\begin{array}{cc}
     1 & 0\\
    \mathbf c^\top & \Tilde{H}^\top 
 \end{array}\right) =0\implies \alpha \left(\begin{array}{cc}
     1+\mathbf d \mathbf c^\top & 0\\
     0 & \Tilde{G}\Tilde{H}^\top 
 \end{array}\right)=0.$$ Since $(\Tilde{\CC},\Tilde{\CD})$ is an LCP of codes (using identity $\pi$, $\sigma$ and $M$ in Proposition~\ref{prop-qq}), from Proposition~\ref{prop-qq}, we get that $rank\left(\Tilde{G}\Tilde{H}^\top\right)=k$. Hence as $rank\left(\Tilde{G}\Tilde{H}^\top\right)=k$ and $ \mathbf d \mathbf c^\top =0$, we have $\alpha=0$ i.e., $\mathbf x = 0$.
\end{proof}
\section{Constacyclic $\FF_q$-linear ACP of codes over $\FF_{q^m}$}\label{sec:consta} 
This section focuses on counting formulas for constacyclic $\FF_q$-linear ACP of codes over $\FF_{q^m}$. Recall that two constacyclic $\FF_q$-linear codes $\CC$ and $\CD$ over $\FF_{q^m}$ form an ACP of codes if $\CC\bigoplus_{\FF_q}\CD=\FF_{q^m}^n$. Let denote $q_i = q^{d_i}$. In the following theorem, we derive necessary and sufficient conditions under which a pair of constacyclic $\FF_q$-linear codes over $\FF_{q^m}$ form an ACP of codes (using identity $\pi$, $\sigma$ and $M$ in Proposition~\ref{prop-qq}).
\begin{thm}\label{th-sqw}
 Let $\CC$ and $\CD$ be two $\lambda$-constacyclic $\FF_q$-linear additive codes over $\FF_{q^m}$, whose Canonical decompositions are given by \eqref{eq-11qq} and \eqref{eq-21qq}, respectively. Then the pair $(\CC, \CD)$ forms an ACP of codes (using identity $\pi$, $\sigma$ and $M$ in Proposition~\ref{prop-qq}) if and only if  $\CC_i \bigoplus_{\mathcal{K}_i} \CD_i=\mathcal{I}_i$, for $0 \leq i \leq r$  and $\CC_{r+k} \bigoplus_{\mathcal{K}_{r+k}} \CD_{r+k} = \mathcal{F}_{r_+k}$, for $1 \leq k \leq s-r$.
\end{thm}
\begin{proof}
The proof can be easily derived using Proposition~\ref{prop-1.1}.   
\end{proof}
In order to study constacyclic $\FF_q$-linear ACP of codes over $\FF_{q^m}$, we first recall that, $|\mathcal{K}_i|=q^{d_i}$, where $|\CC^{q}_{1+ti}|=d_i$ for all $0 \leq i \leq s$. Further $|\mathcal{I}_i|=\left(q^m\right)^{d_i}=|\mathcal{K}_i|^{m}$, for all $0\leq i \leq r$. Also, $|\mathcal{F}_{r+k}|=\left(q^m\right)^{d_{r+k}}=|\mathcal{K}_{r+k}|^{m}$, where $|\CC^{q}_{1+t(r+k)}|=d_{r+k}$ for all $1 \leq r+k \leq s-r$. This implies that $\mathcal{I}_i$ is an $\mathcal{K}_i$-linear space of $\mathcal{K}_i^m$ for $0 \leq i \leq r$ and $\mathcal{F}_{k+r}$ is an $\mathcal{K}_{k+r}$-linear space of $\mathcal{K}_{r+k}^m$ for $1 \leq k \leq s-r$. It is well known that the number of $\FF_{q^{d_i}}$-subspace of $\FF_{q^{d_i}}^m$ is equal to $\mathlarger{\sum}\limits_{v=0}^m \left[ \begin{array}{cc}
     m  \\
     v 
\end{array}\right]_{{q_i}}=\dfrac{(q_i^m-1)(q_i^m-q_i)\cdots(q_i^m-q_i^{k-1})}{(q_i^v-1)(q_i^v-q_i)\cdots(q_i^v-q_i^{v-1})}$.
From the above discussion, we have the following theorem.
\begin{thm}\label{asa}
 The number of $\lambda$-constacyclic ACP of codes over $\FF_{q^m}$ is equal to   
 $$\mathlarger{\prod}_{i=1}^s\left(\mathlarger{\sum}\limits_{v=0}^m \left[ \begin{array}{cc}
     m  \\
     v 
\end{array}\right]_{{q_i}}q_i^{v(m-v)}\right).$$
\end{thm}
\begin{proof}
By Theorem~\ref{th-sqw}, the pair $(\CC, \CD)$ is an ACP if and only if $\CC_i \bigoplus_{\mathcal{K}_i} \CD_i=\mathcal{I}_i$, for $0 \leq i \leq r$  and $\CC_{r+k} \bigoplus_{\mathcal{K}_{r+k}} \CD_{r+k} = \mathcal{F}_{r_+k}$ for $1 \leq k \leq s-r$. Equivalently, $(\CC, \CD)$ is an ACP if and only if  $(\CC_i, \CD_i)$ is an LCP for all $0 \leq i \leq s$. If $(\CC_i, \CD_i)$ form an LCP for $0 \leq i \leq s$, then the number of choices of $\CC_i$ is equal to $$\mathlarger{\sum}\limits_{v=0}^m \left[ \begin{array}{cc}
     m  \\
     v 
\end{array}\right]_{{q_i}}$$ and the number of choices of $\CD_i$ is equal to 
$$\dfrac{(q_i^m-q_i^{v})(q_i^m-q_i^{v+1})\cdots(q_i^m-q_i^{m-1})}{(q_i^{m-v}-1)(q_i^{m-v}-q_i)\cdots(q_i^{m-v}-q_i^{m-v-1})}=q_i^{m-v}.$$
Therefore, the number of ACP of codes $(\CC, \CD)$ is
$$\# (\CC, \CD) = \mathlarger{\prod}_{i=1}^s\left(\mathlarger{\sum}\limits_{v=0}^m \left[ \begin{array}{cc}
     m  \\
     v 
\end{array}\right]_{{q_i}}q_i^{v(m-v)}\right).$$ This completes the proof.
\end{proof}
We present a numerical example as follows.
\begin{example}
Take, $q=3^2$ and consider cyclic $\FF_3$-linear additive code over $\FF_{3^2}$ of length $10$. In this special case, we have $p=3, m=2, \lambda=1$ and $n=10$. Next, $X^{10}-1=\prod_{i=0}^3 p_i(X)$, where
\begin{align*}
& p_0(X)=X-1=2+X&\\
& p_1(X)=1+X&\\
& p_2(X)=1+X+X^2+X^3+X^4&\\
& p_3(X)=1+2X+X^2+2X^3+X^4&
\end{align*}
are irreducible polynomials in $\FF_3[X]$. By Theorem~\ref{asa}, we deduce that the number of  cyclic $\FF_3$-linear ACP of codes over $\FF_{3^2}$ of length $30$
equals 
$$\# (\CC, \CD) = \mathlarger{\prod}_{i=0}^3\left(\mathlarger{\sum}\limits_{v=0}^2 \left[ \begin{array}{cc}
2  \\ v 
\end{array}\right]_{{3^{d_i}}}q_i^{v(2-v)}\right)=\mathlarger{\prod}_{i=0}^3\left(2+2q_i\right)=4,40,896.$$
\end{example}

In Table~\ref{table1.0}, we present some examples of LCD codes corresponding to additive code from Theorem~\ref{th-1zx}.

\begin{table}[h]  \centering
\begin{tabular}{|c|c|c|c|} \hline
No. & Generator of additive codes  & Trace Codes  & Remark   \\  
& & $Tr(\CC):=[n, k, d_H]_2$  & \\
\hline
1 &  $\left(\begin{array}{cccccc}
  \omega & \omega^2 & 0 & \omega & 0 \\
  0      & \omega   & \omega^2 & 0 & \omega 
\end{array}\right)$ 
& $[5,~2,~3]$ & not LCD code but optimal    \\  \hline
2 &  $\left(\begin{array}{ccccccc}
  \omega & \omega^2 & 0 & \omega & \omega & 0 \\
  0      & \omega   & \omega^2 & 0 &\omega &\omega 
\end{array}\right)$ 
& $[6,~2,~4]$ & not LCD code but optimal    \\  \hline
3 &  $\left(\begin{array}{ccccccccc}
  \omega & 0 &\omega^2 & 0 & \omega & \omega & 0 & 0 \\
  0      & \omega & 0 &\omega^2 & 0 & \omega & \omega & 0\\
  0 & 0 & \omega & 0 &\omega^2 & 0 & \omega & \omega 
\end{array}\right)$ 
& $[8,~3,~4]$ & not LCD code but optimal    \\  \hline
4 &  $\left(\begin{array}{ccccccccc}
  \omega & 0 &\omega^2 & 0 & \omega & \omega & 0 & 0 \\
  0      & \omega & 0 &\omega^2 & 0 & \omega & \omega & 0\\
  0 & 0 & \omega & 0 & 0  & 0 & \omega & \omega 
\end{array}\right)$ 
& $[8,~3,~3]$ & an LCD code but  not optimal   \\  \hline
5 &  $\left(\begin{array}{ccccccccc}
  \omega &\omega & 0 & \omega & \omega & 0 &\omega &\omega & 0\\
  0      & \omega &\omega & 0 & \omega & \omega & 0 &\omega &\omega
\end{array}\right)$ 
& $[9,~2,~6]$ & an LCD code and optimal    \\  \hline
6 &  $\left(\begin{array}{ccccccccc}
  \omega &\omega^2 & 0 & \omega & 0 &\omega &\omega & 0 & 0\\
  0      & \omega &\omega^2 & 0 & \omega & 0 &\omega &\omega & 0\\
  0 & 0 & \omega &\omega & 0 & \omega^2 & 0 & \omega^2 & \omega
\end{array}\right)$ 
& $[9,~3,~3]$ & an LCD code and optimal    \\  \hline
7 &  $\left(\begin{array}{ccccccccc}
  \omega & 1 & 0 & \omega & 0 &\omega &\omega & 0 & 0\\
  0      & \omega & 1 & 0 & \omega & 0 &\omega &\omega & 0\\
  0 & 0 & \omega &\omega & 0 & \omega^2 & 0 & \omega^2 & \omega
\end{array}\right)$ 
& $[9,~3,~4]$ & an LCD code and optimal   \\  \hline
8 &  $\left(\begin{array}{cccccccccc}
  \omega & \omega & 0 & \omega & \omega & 0 &\omega &\omega^2 & \omega & 0\\
  0      & \omega & \omega & 0 & \omega & \omega & 0 &\omega &\omega^2 & \omega
\end{array}\right)$ 
& $[10,~2,~6]$ & an LCD code and optimal    \\  \hline
9 &  $\left(\begin{array}{cccccccccc}
  \omega & \omega & 0 & 0 & \omega & 0 &\omega &\omega & 0 & 0\\
  0      & \omega & \omega & 0 & 0 & \omega^2 & 0 &\omega &\omega & 0\\
  0& 0      & \omega & \omega & 0 & 0 & \omega & 0 &\omega^2 &\omega
\end{array}\right)$ 
& $[10,~3,~5]$ & not LCD code and optimal    \\  \hline
10 &  $\left(\begin{array}{cccccccccc}
  \omega & \omega & 0 & 0 & \omega & 0 &\omega & 0 & 0 & 0\\
  0      & \omega & \omega & 0 & 0 & \omega & 0 &\omega & 0 & 0\\
  0& 0      & \omega & \omega & 0 & 0 & \omega^2 & 0 &\omega & 0\\
  0& 0& 0      & \omega & \omega & 0 & 0 & \omega & 0 &\omega
\end{array}\right)$ 
& $[10,~4,~4]$ & an LCD code and optimal    \\  \hline
11 &  $\left(\begin{array}{ccccccccccc}
  \omega & \omega & 0 & \omega & \omega & 0 &\omega & \omega & 0 & \omega & 0\\
  0      & \omega & \omega & 0 & \omega & \omega & 0 &\omega & \omega & 0 & \omega
\end{array}\right)$ 
& $[11,~2,~7]$ & not LCD code and optimal    \\  \hline
12 &  $\left(\begin{array}{ccccccccccc}
  \omega & \omega & 0 & 0 & \omega & \omega & 0 & 0 &\omega & 0 & 0\\
  0      & \omega & \omega & 0 & 0 & \omega & \omega & 0 & 0 &\omega & 0\\
  0 & 0 & \omega & \omega & 0 & 0 & \omega & \omega^2 & 0 & 0 &\omega
\end{array}\right)$ 
& $[11,~3,~5]$ & an LCD code and optimal    \\  \hline
\end{tabular}
 \caption{Some LCD codes }\label{table1.0}
 \end{table}

\section{Conclusion}\label{sec:con}
The main contribution of this paper is that a very general type of ACP codes is characterized and constructed. Moreover, we obtained a condition for an additive pair of codes to be an ACP of codes. Furthermore, we provide a necessary and sufficient condition for an ACP of codes. Finally, we obtained an ACP of codes from given linear complementary pairs of codes.

\end{document}